\documentclass[runningheads]{article}
\usepackage{fullpage}           

\usepackage{graphicx}
\usepackage{amsmath}
\usepackage{amsthm}             
\usepackage{amssymb}
\usepackage{color}
\usepackage{hyperref}
\newtheorem{theorem}{Theorem}
\newtheorem{lemma}[theorem]{Lemma}
\newtheorem{corollary}[theorem]{Corollary}

\newtheorem{problem}[theorem]{Problem}

\newtheorem{hypothesis}{Hypothesis}

\newcommand{\descr}[1]{\emph{\textbf{(#1).}}}

\newcommand{\eps}{\varepsilon}
\DeclareMathOperator{\poly}{poly}
\newcommand{\Fr}{Fr\'echet }

\begin{document}
\title{Fine-Grained Complexity Theory: \\Conditional Lower Bounds for Computational Geometry}
%
%
\author{Karl Bringmann\thanks{Saarland University and Max Planck Institute for Informatics,  Saarland Informatics Campus, Germany,
\url{bringmann@cs.uni-saarland.de} This work is part of the project TIPEA that has received funding from the European Research Council (ERC) under the European Unions Horizon 2020 research and innovation programme (grant agreement No.\ 850979).}}
%
%
%
\maketitle              
\begin{abstract}

Fine-grained complexity theory is the area of theoretical computer science that proves conditional lower bounds based on the Strong Exponential Time Hypothesis and similar conjectures. This area has been thriving in the last decade, leading to conditionally best-possible algorithms for a wide variety of problems on graphs, strings, numbers etc. 

This article is an introduction to fine-grained lower bounds in computational geometry, with a focus on lower bounds for polynomial-time problems based on the Orthogonal Vectors Hypothesis. Specifically, we discuss conditional lower bounds for nearest neighbor search under the Euclidean distance and \Fr distance.

\end{abstract}

\section{Introduction}

The term \emph{fine-grained complexity theory} was coined in the last decade to describe the area of theoretical computer science that proves conditional lower bounds on the time complexity of algorithmic problems, assuming some hypothesis. The goal is to explain the computational complexity of many different problems based on a small number of core barriers. The general approach dates back to the introduction of 3SUM-hardness in '95~\cite{GajentaanO95} (or even to the introduction of NP-hardness, depending on the interpretation). The last decade has seen several new hypotheses and a wealth of new techniques for proving conditional lower bounds, leading to a large body of literature on the topic, see the surveys~\cite{williams2018some,Bringmann19}.
In this article we give a self-contained introduction to recent fine-grained complexity results in the area of computational geometry. Instead of the most technically deep results, we focus on simple techniques that can be easily transferred to other problems. Moreover, we focus on lower bounds for polynomial-time problems.

\medskip
The basic setup of fine-grained lower bounds is similar to classic NP-hardness reductions: A fine-grained reduction from problem $P$ to problem $Q$ is an algorithm that given an instance $I$ of size $n$ for problem $P$ computes in time $r(n)$ an equivalent instance $J$ of size $s(n)$ for problem $Q$.\footnote{What we sketched here is a \emph{many-one} reduction, since each instance of $P$ is reduced to one instance of $Q$. One can also consider \emph{Turing} reductions, where the reduction algorithm is allowed to make several calls to an oracle for $Q$. See \cite[Definition~1]{CarmosinoGIMPS16} for the formal definition of (Turing-style) fine-grained reductions.} Thus, if there is an algorithm solving problem $Q$ in time $T(n)$, by this reduction there is an algorithm solving problem $P$ in time $r(n) + T(s(n))$. In particular, if $r(n) + T(s(n))$ is faster than the hypothesized optimal time complexity of problem $P$, then problem $Q$ cannot be solved in time $T(n)$ assuming the hypothesis for $P$. We will see several concrete examples of this argumentation throughout this article. 


\subsection{Hardness Hypotheses}

Let us discuss the three main hypotheses used in computational geometry.

\subsubsection{3SUM Hypothesis}

In the 3SUM problem, given $n$ integers, we want to decide whether any three of them sum to 0. The 3SUM Hypothesis postulates that the classic $O(n^2)$-time algorithm for 3SUM cannot be improved to time $O(n^{2-\eps})$ for any $\eps > 0$. This hypothesis was introduced in '95 in a seminal work by Gajentaan and Overmars~\cite{GajentaanO95}, making computational geometry a pioneer in fine-grained complexity theory.
We refer to \cite{williams2018some} for an overview of lower bounds based on the 3SUM Hypothesis; in this introduction we focus on other hypotheses.

\subsubsection{Strong Exponential Time Hypothesis}

The strongest new impulse for conditional lower bounds in the last two decades was the introduction of the Strong Exponential Time Hypothesis.
This hypothesis concerns the fundamental $k$-SAT problem: Given a formula $\phi$ in conjunctive normal form of width $k$ on $n$ variables and $m$ clauses, decide whether $\phi$ is satisfiable. Naively this problem can be solved in time $O(2^n m)$. Improved algorithms solve $k$-SAT in time $O(2^{(1-\eps_k)n})$ for some constant $\eps_k > 0$, but for all known algorithms the constant $\eps_k$ tends to~0 for $k \to \infty$. This lead Impagliazzo and Paturi~\cite{ImpagliazzoP01} to postulate the following:

\begin{hypothesis} \descr{Strong Exponential Time Hypothesis -- SETH}
  For any $\eps > 0$, there exists $k \ge 3$ such that $k$-SAT on formulas with $n$ variables cannot be solved in time $O(2^{(1-\eps)n})$.
\end{hypothesis}

This has become the most standard hypothesis in fine-grained complexity theory~\cite{williams2018some}, and it has been used to prove tight lower bounds for a wide variety of problems, see, e.g.,~\cite{AbboudBW15,AbboudW14,Bringmann14,%
BringmannK15,BringmannK18,BringmannKN19,BringmannM16,%
BringmannN21,buchin2016fine,%
BuchinOS19,Chen20,Rubinstein18,Williams18}.

%
%

\subsubsection{Orthogonal Vectors Hypothesis}

In the Orthogonal Vectors problem (OV), given sets of Boolean vector $A,B \subseteq \{0,1\}^d$ of size $n$, we ask whether there exists a pair $(a,b) \in A \times B$ that is orthogonal, that is, $\langle a,b \rangle = \sum_{i=1}^d a_i \cdot b_i = 0$. Naively this problem can be solved in time $O(n^2 d)$. For small dimension $d = O(\log n)$ there are improved algorithms~\cite{AbboudWY15}, but for $\omega(\log n) \le d \le n^{o(1)}$ no algorithm running in time $O(n^{2-\eps})$ is known. 
This barrier is formalized as follows.

\begin{hypothesis} \descr{OV Hypothesis -- OVH~\cite{Williams05}} \label{hyp:OV}
  For any $\eps > 0$, OV cannot be solved in time $O(n^{2-\eps} \poly(d))$.
\end{hypothesis}

Note that for $d = n^{\Omega(1)}$ we can naively solve OV in time $O(n^2 d) = \poly(d) = O(n^{2-\eps} \poly(d))$, and thus OVH does not apply. Indeed, the hypothesis only asserts that there \emph{exists} a dimension $d = d(n)$ such that OV cannot be solved in time $O(n^{2-\eps} \poly(d))$; this dimension $d$ must be of the form $\omega(\log n) \le d \le n^{o(1)}$. 

OVH has been used to prove tight conditional lower bounds for a wide range of problems, see, e.g.,~\cite{AbboudBW15,Bringmann14,BringmannK15,%
BringmannK18,BringmannM16,BringmannN21,BuchinOS19,Williams18}.
It is known that OVH is at least as believable as SETH, because SETH implies OVH~\cite{Williams05}. 

%

\medskip
In this article we focus on lower bounds based on OVH (since SETH implies OVH this also yields lower bounds based on SETH). Specifically, in Section~\ref{sec:nearest} we consider nearest neighbor search, and in Section~\ref{sec:frechet} we discuss curve similarity.

\section{Nearest Neighbor Search}
\label{sec:nearest}

A fundamental problem of computer science is to compute the nearest neighbor of a point $q \in \mathbb{R}^d$ among a set of points $P \subset \mathbb{R}^d$, that is, to determine the point $p \in P$ minimizing the Euclidean distance $\|p-q\|$.
This has an abundance of applications such as pattern recognition, spell checking, or coding theory. 
These applications often come in the form of a data structure problem, where we can first preprocess $P$ to build a data structure that can then quickly answer nearest neighbor queries. 
Naively, a nearest neighbor query can be answered in time $O(nd)$, where $n$ is the number of points in the data set $P$. Improved algorithms exist in small dimensions, for example k-d-trees have a worst-case query time of $O(d\cdot n^{1-1/d})$~\cite{LeeW77}.
However, already for a large constant dimension $d \ge 1/\eps$ this query time is essentially linear, specifically it is $\Omega(n^{1-\eps})$. We can thus ask: 

\begin{center}
  \emph{Does high-dimensional nearest neighbor search require near-linear query time?}
\end{center}

In the following we answer this question affirmatively assuming OVH.
To connect nearest neighbor search to the OV problem we make use of the following embedding, which maps Boolean vectors to points in $\mathbb{R}^d$ such that from the points' Euclidean distance we can read off whether the vectors are orthogonal.

\begin{lemma}[Embedding Orthogonality into Euclidean Distance] \label{lem:embedEuclid}
  There are functions $\mathcal{A}, \mathcal{B} \colon \{0,1\}^d \mapsto \mathbb{R}^d$ and a threshold $\tau$ such that $\langle a,b \rangle = 0$ if and only if $\| \mathcal{A}(a) - \mathcal{B}(b) \| \le \tau$ for any $a,b \in \{0,1\}^d$. The functions $\mathcal{A},\mathcal{B}$ and the threshold $\tau$ can be evaluated in time $O(d)$.
\end{lemma}
\begin{proof}
  For any $a \in \{0,1\}^d$ we construct $p := \mathcal{A}(a)$ by setting $p_i := 1 + 2 a_i$ for any $1 \le i \le d$. Similarly, for any $b \in \{0,1\}^d$ we construct $q := \mathcal{B}(b)$ by setting $q_i := 2 - 2b_i$. Note that $|p_i - q_i| = |2(a_i+b_i)-1|$, which evaluates to 3 if $a_i=b_i=1$ and to 1 otherwise. Therefore, we obtain
  \[ \|p-q\| = \Big(\sum_{i=1}^d |p_i-q_i|^2\Big)^{1/2} = \big(3^2 \cdot \langle a,b \rangle + 1^2 \cdot (d - \langle a,b \rangle)\big)^{1/2} = (d + 8 \langle a,b \rangle)^{1/2}. \]
  Setting $\tau := d^{1/2}$ yields $\|p-q\| = \|\mathcal{A}(a) - \mathcal{B}(b)\| \le \tau$ if and only if $\langle a,b \rangle = 0$.
\end{proof}

\subsection{Bichromatic Closest Pair}

We use the above embedding to prove a conditional lower bound for the the Bichromatic Closest Pair problem, an offline variant of nearest neighbor search:

\begin{problem}[Bichromatic Closest Pair]
  Given sets $P,Q \subset \mathbb{R}^d$ of size $n$, compute the pair $(p,q) \in P \times Q$ minimizing the Euclidean distance $\|p-q\|$. 
\end{problem}

\noindent
Bichromatic Closest Pair cannot be solved in time $O(n^{2-\eps} \poly(d))$ under OVH:

\begin{theorem}[Lower Bound for Bichromatic Closest Pair~\cite{AlmanW15}] \label{thm:bcp}
  For any $\eps > 0$, Bichromatic Closest Pair cannot be solved in time $O(n^{2-\eps} \poly(d))$, unless OVH fails. 
\end{theorem}
\begin{proof}
  We reduce from OV to Bichromatic Closest Pair using the embedding from Lemma~\ref{lem:embedEuclid}. Given an OV instance $(A,B)$ of size $n$ in dimension $d$, we construct the point sets $P = \{\mathcal{A}(a) \mid a \in A\}$ and $Q := \{\mathcal{B}(b) \mid b \in B\}$. By Lemma~\ref{lem:embedEuclid}, the bichromatic closest pair of $(P,Q)$ has distance $\le \tau$ if and only if there exists an orthogonal pair of vectors. Thus, a solution to the constructed Bichromatic Closest Pair instance solves the given OV instance. Since $n$ and $d$ do not change, the running time lower bound is immediate from OVH (Hypothesis~\ref{hyp:OV}).
\end{proof}

\subsection{Nearest Neighbor Data Structures}

Now we consider the data structure version of nearest neighbor search.

\begin{problem}[Nearest Neighbor Data Structure]
  In the preprocessing we are given a set $P \subset \mathbb{R}^d$ of size $n$ and we build a data structure. The data structure allows to answer nearest neighbor queries: Given a point $q \in \mathbb{R}^d$, compute the point $p \in P$ minimizing the Euclidean distance $\|p-q\|$.
\end{problem}

Observe that any nearest neighbor data structure also solves the Bichromatic Closest Pair Problem, by building the data structure for $P$ and then querying every $q \in Q$. If the data structure has preprocessing time $T_P(n,d)$ and query time $T_Q(n,d)$, then this solves Bichromatic Closest Pair in time $T_P(n,d) + n\cdot T_Q(n,d)$. 
Theorem~\ref{thm:bcp} thus implies that Bichromatic Closest Pair cannot be solved with preprocessing time $O(n^{2-\eps} \poly(d))$ and query time $O(n^{1-\eps} \poly(d))$:

\begin{corollary}[Lower Bound for Nearest Neighbor Data Structures I]
  For any $\eps > 0$, there is no nearest neighbor data structure with preprocessing time $O(n^{2-\eps} \poly(d))$ and query time $O(n^{1-\eps} \poly(d))$, unless OVH fails.
\end{corollary}

It might seem natural that the preprocessing time is limited to $O(n^{2-\eps})$, because from OVH we can prove only quadratic lower bounds. In the following we show that this intuition is wrong, and the above corollary can be improved to rule out \emph{any polynomial} preprocessing time. To this end, we need an unbalanced version of OVH, which shows that the brute force enumeration of all $|A| \cdot |B|$ pairs of vectors is also necessary when $|A| = n^\alpha \ll n = |B|$.
This tool was introduced in~\cite{AbboudW14}; see also~\cite{BringmannK18} for a proof that unbalanced OV and standard OV are equivalent.

\begin{lemma}[\cite{AbboudW14}] \label{lem:UOVH}
  For any $\eps, \alpha \in (0,1)$, OV on instances $(A,B)$ with $|B| = n$ and $|A| = \Theta(n^\alpha)$ cannot be solved in time $O(n^{1+\alpha-\eps} \poly(d))$, unless OVH fails.
\end{lemma}
\begin{proof}
  Let $(A',B')$ be a balanced instance of OV, that is, $|A'|=|B'|=n$. Split $B'$ into $\Theta(n^{1-\alpha})$ sets $B'_1,\ldots,B'_\ell$ of size $\Theta(n^\alpha)$. Run an unbalanced OV algorithm on each pair $(A',B'_i)$, and note that from the results we can infer whether $(A',B')$ contains an orthogonal pair of vectors.
  If each unbalanced instance can be solved in time $O(n^{1+\alpha-\eps})$, then all $\Theta(n^{1-\alpha})$ unbalanced instances in total can be solved in time $O(n^{2-\eps})$, contradicting OVH.
\end{proof}

With this tool, we can rule out any polynomial preprocessing time $\poly(n,d)$ and query time $O(n^{1-\eps} \poly(d))$ for nearest neighbor search:

\begin{theorem}[Lower Bound for Nearest Neighbor Data Structures II]
  For any $\eps, \beta > 0$, there is no nearest neighbor data structure with preprocessing time $O(n^{\beta} \poly(d))$ and query time $O(n^{1-\eps} \poly(d))$, unless OVH fails.
\end{theorem}
\begin{proof}
  Fix $\eps, \beta > 0$ and suppose nearest neighbor can be solved with preprocessing time $O(|P|^{\beta} \poly(d))$ and query time $O(|P|^{1-\eps} \poly(d))$.
  Set $\alpha := 1/\beta$.
  Given an OV instance $(A,B)$ with $|B| = n$ and $|A| = \Theta(n^\alpha)$, we use the embedding from Lemma~\ref{lem:embedEuclid} to construct the sets $P := \{\mathcal{A}(a) \mid a \in A\}$ and $Q := \{\mathcal{B}(b) \mid b \in B\}$. We run the preprocessing of the nearest neighbor data structure on $P$; this takes time $O(|P|^\beta \poly(d)) = O((n^\alpha)^\beta \poly(d)) = O(n \poly(d))$. Then we query the data structure for each $q \in Q$; over all $|Q|$ queries this takes total time 
  \[ O(|Q| \cdot |P|^{1-\eps}\poly(d)) = O(n^{1+\alpha \cdot (1-\eps)} \poly(d)) = O(n^{1+\alpha-\eps'} \poly(d)), \]
  for $\eps' := \alpha \cdot \eps$. By Lemma~\ref{lem:embedEuclid}, some query $q \in Q$ returns a point $p \in P$ within distance $\tau$ if and only if there exists an orthogonal pair of vectors in $A \times B$. We can thus solve unbalanced OV in time $O(n^{1+\alpha-\eps'} \poly(d))$, contradicting Lemma~\ref{lem:UOVH}.
\end{proof}

We have thus shown that high-dimensional nearest neighbor search requires almost-linear query time, even if we allow any polynomial preprocessing time. 

\subsection{Further Results on Nearest Neighbor Search}

Let us discuss some advanced research directions on nearest neighbor search. The proofs here are beyond the scope of this introduction to the topic.

\begin{itemize}
\item \emph{Smaller Dimension:}
The best known query time for nearest neighbor search is of the form $n^{1-\Theta(1/d)}$~\cite{LeeW77}, which is near-linear $n^{1-o(1)}$ for any unbounded dimension $d = \omega(1)$. Recall that OVH asserts hardness for some dimension $\omega(\log n) \le d \le n^{o(1)}$. 
A line of research has tried to close this gap~\cite{Williams18,Chen20}; the current record shows that Theorem~\ref{thm:bcp} already holds in dimension $d = 2^{O(\log^* n)}$~\cite{Chen20}. It remains an important open problem to close the remaining gap and show hardness for any dimension $d = \omega(1)$.

\item \emph{Approximate Nearest Neighbor:}
In many practical applications it suffices to compute nearest neighbors approximately. Note that the OV problem asks whether there is a pair of vectors with $\langle a,b \rangle = 0$ or whether all vectors have $\langle a,b \rangle \ge 1$. Inspecting the proof of Lemma~\ref{lem:embedEuclid}, we see that it is hard to distinguish between Euclidean distance at most $d^{1/2}$ or at least $(d+8)^{1/2}$. This shows hardness of computing a $(d+8)^{1/2} / d^{1/2} = 1 + \Theta(1/d)$ approximation for Bichromatic Closest Pair. 
A big leap forward was made by Rubinstein~\cite{Rubinstein18}, who proved that Theorem~\ref{thm:bcp} even holds for $(1+\delta)$-approximation algorithms, where $\delta = \delta(\eps)$ is some positive constant. 
See also~\cite{RubinsteinW19} for more hardness of approximation results in fine-grained complexity theory.
\end{itemize}

%

\section{Curve Similarity and the \Fr Distance}
\label{sec:frechet}

We now turn to a different realm of applications.
For our purposes, a \emph{curve} is a sequence of points in the plane, that is, $\pi = (\pi_1,\ldots,\pi_n)$ with $\pi_i \in \mathbb{R}^2$. We call $n$ the \emph{length} of $\pi$.
A typical task is to judge the similarity of two given curves. Several distance measures have been proposed for this task, but the most classical and most popular in computational geometry is the \Fr distance\footnote{For simplicity, we focus on the \emph{discrete} \Fr distance~\cite{EiterMannila94} instead of the slightly more standard \emph{continuous} variant~\cite{AltG95}.}. For intuition, imagine a dog walking along curve $\pi$ and its owner walking along curve $\sigma$, connected by a leash. They start at the respective startpoints and end at their endpoints, and at any point in time either the dog advances to the next vertex along its curve, or the owner advances, or they both advance together. The shortest possible leash length admitting such a traversal is called the \Fr distance of $\pi$ and $\sigma$.

Formally, for curves $\pi=(\pi_1,\ldots,\pi_n)$ and $\sigma = (\sigma_1,\ldots,\sigma_m)$, a \emph{traversal} is a sequence $((i_1,j_1),\ldots,(i_T,j_T))$ such that $(i_1,j_1) = (1,1)$, $(i_T,j_T) = (n,m)$, and for every $1 \le t < T$ we have $(i_{t+1},j_{t+1}) \in \{(i_t+1,j_t), (i_t,j_t+1), (i_t+1,j_t+1)\}$. The (discrete) \Fr distance between $\pi$ and $\sigma$ is defined as
\[ d_F(\pi,\sigma) = \min_{((i_1,j_1),\ldots,(i_T,j_T))} \max_{1 \le t \le T} \| \pi_{i_t} - \sigma_{j_t} \|, \]
where the minimum goes over all traversals of $\pi$ and $\sigma$. 

The \Fr distance of two curves of length $n$ can be computed in time $O(n^2)$, by a simple dynamic programming algorithm that computes the \Fr distance of any prefix $(\pi_1,\ldots,\pi_i)$ of $\pi$ and any prefix $(\sigma_1,\ldots,\sigma_j)$ of $\sigma$~\cite{EiterMannila94}.

\smallskip
In the following, we first discuss the \Fr distance from the viewpoint of nearest neighbor search, and then we elaborate on the problem of computing the \Fr distance of two given curves.

\subsection{Nearest Neighbor Search under \Fr Distance}

We start with an embedding of vectors into curves, similar to Lemma~\ref{lem:embedEuclid}.

\begin{lemma}[Embedding Orthogonality into \Fr Distance~\cite{Bringmann14}] \label{lem:embedFr}
  There are functions $\mathcal{A}, \mathcal{B}$ mapping any $z \in \{0,1\}^d$ to a curve of length $d$ in the plane, such that $\langle a,b \rangle = 0$ if and only if $d_F(\mathcal{A}(a), \mathcal{B}(b)) \le 1$ for any $a,b \in \{0,1\}^d$. The functions $\mathcal{A},\mathcal{B}$ can be evaluated in time $O(d)$.
\end{lemma}
\begin{proof}
  For any $a \in \{0,1\}^d$ we construct the curve $\pi := \mathcal{A}(a)$ by setting $\pi_i := (3i, 1 + 2 a_i) \in \mathbb{R}^2$ for any $1 \le i \le d$. Similarly, for any $b \in \{0,1\}^d$ we construct $\sigma := \mathcal{B}(b)$ by setting $\sigma_i := (3i, 2 - 2b_i)$. Note that $\|\pi_i - \sigma_i\| = |2(a_i+b_i)-1|$, which evaluates to 3 if $a_i=b_i=1$ and to 1 otherwise. Moreover, for $i \ne j$ we have $\|\pi_i - \sigma_j\| \ge 3$. 
  
  Consider a traversal of $\pi$ and $\sigma$. If at some point the dog advances but not the owner (or the owner advances but not the dog), we get a distance of the form $\|\pi_i - \sigma_j\|$ for $i \ne j$, and thus the leash length must be at least 3. In the remaining case, the dog and its owner always advance together, meaning that at time $i$ the dog is at position $\pi_i$ and the owner is at position $\sigma_i$. This traversal has distance $\max_{1 \le i \le d} \|\pi_i - \sigma_i\| = \max_{1\le i \le d} |2(a_i+b_i)-1|$, which is 1 if $a,b$ are orthogonal, and 3 otherwise. Hence, $d_F(\pi,\sigma) \le 1$ holds if and only if $\langle a,b \rangle = 0$.
\end{proof}

Using this embedding, we can show lower bounds for nearest neighbor search among curves in the plane, analogously to the results for Euclidean nearest neighbor search from Section~\ref{sec:nearest} (the same proofs work almost verbatim). 
Specifically, in the problem Bichromatic Closest Pair under \Fr Distance we are given sets $P,Q$, each containing $n$ curves of length $d$ in the plane, and we want to compute the pair $(\pi,\sigma) \in P \times Q$ that minimizes the \Fr distance $d_F(\pi,\sigma)$. Naively, this can be solved in time $O(n^2 d^2)$.

\begin{theorem}[Lower Bound for Bichromatic Closest Pair under \Fr Distance] \label{thm:bcpfrechet}
  For any $\eps > 0$, Bichromatic Closest Pair under \Fr Distance cannot be solved in time $O(n^{2-\eps} \poly(d))$, unless OVH fails. 
\end{theorem}

Similarly, in nearest neighbor data structures for the \Fr distance we can preprocess a given set $P$ consisting of $n$ curves of length $d$ in the plane, and then given a query curve $\sigma$ of length $d$ in the plane we want to find the curve $\pi \in P$ minimizing $d_F(\pi,\sigma)$.

\begin{theorem}[Lower Bound for Nearest Neighbor Data Structures under \Fr Distance] \label{thm:nnfrechet}
  For any $\eps, \beta > 0$, there is no data structure for nearest neighbor search under \Fr distance with preprocessing time $O(n^{\beta} \poly(d))$ and query time $O(n^{1-\eps} \poly(d))$, unless OVH fails.
\end{theorem}

\subsection{Computing the \Fr Distance}


A classic dynamic programming algorithm computes the \Fr distance between two curves of length $n$ in time $O(n^2)$~\cite{EiterMannila94}. A breakthrough result from '14 shows a tight lower bound ruling out time $O(n^{2-\eps})$ under OVH~\cite{Bringmann14}. This result paved the way for tight lower bounds for many other dynamic programming problems (mostly outside of computational geometry, see, e.g.,~\cite{AbboudBW15,BringmannK15}). Here we give a very brief sketch of this result.

\begin{theorem}[Lower Bound for \Fr Distance~\cite{Bringmann14}] \label{thm:frechethardness}
  For any $\eps > 0$, the \Fr distance cannot be computed in time $O(n^{2-\eps})$.
\end{theorem}

\noindent
\emph{Proof Sketch.} 
Given an OV instance $(A,B)$ on $n$ vectors in dimension $d$, we construct two curves $\pi,\sigma$ of length $N = O(nd)$ such that $d_F(\pi,\sigma) \le 1$ if and only if $(A,B)$ contains an orthogonal pair. It then follows that if the \Fr distance can be computed in time $O(N^{2-\eps})$, then OV can be solved in time $O((nd)^{2-\eps}) = O(n^{2-\eps} \poly(d))$, contradicting OVH (Hypothesis~\ref{hyp:OV}). 

To construct the curves $\pi,\sigma$, we start with \emph{vector gadgets}. These gadgets are similar to the embedding in Lemma~\ref{lem:embedFr}, but they are restricted to a much smaller region in space. Specifically, for each vector $a \in A$ we construct a curve $VG(a)$ as the sequence of points 
$( (-1)^i \delta, 0.5 - (-1)^{a_i} \delta^2 ) \in \mathbb{R}^2$ for $1 \le i \le d$, where $\delta > 0$ is a small constant.
Similarly, for each vector $b \in B$ we construct a curve $VG(b)$ as the sequence of points $( (-1)^i \delta, -0.5 + (-1)^{b_i} \delta^2 )$ for $1 \le i \le d$. Analogously to Lemma~\ref{lem:embedFr}, we can show that $d_F(VG(a),VG(b)) \le 1$ if and only if $\langle a,b \rangle = 0$.

The final and most complicated step of the reduction is the \emph{OR gadget}. This gadget combines the curves $VG(a)$ for all $a \in A$ into one curve $\pi$, and similarly it combines the curves $VG(b)$ for all $b \in B$ into one curve $\sigma$, such that $d_F(\pi,\sigma) \le 1$ if and only if there exist $a \in A, b\in B$ with $d_F(VG(a),VG(b)) \le 1$. To this end, we introduce auxiliary points at the following positions:
\[ s = (-0.5, 0),\; t = (0.5,0),\; s^* = (-0.5,-1),\; t^* = (0.5, 1). \]
The final curve $\pi$ repeats the pattern $(s, VG(a), t)$ for all $a \in A$. 
The final curve~$\sigma$ starts with $s$ and $s^*$, then walks through all vector gadgets $VG(b)$, and ends with $t^*$ and $t$. 
One can show that these curves satisfy 
$d_F(\pi,\sigma) \le 1$ if and only if $(A,B)$ contains an orthogonal pair, for details see~\cite{Bringmann14}.


%

\subsection{Further Results on \Fr Distance}

\begin{itemize}
\item \emph{Robustness:}
For reductions to geometric problems a common concern is the precision needed to write down the constructed instances. The reductions shown in this article are very robust: they only require $O(\log d)$-bit coordinates, and some can even be made to work with $O(1)$-bit coordinates.

\item \emph{Hardness of Approximation:}
Inspecting the proof of Lemma~\ref{lem:embedFr}, we see that it is hard to distinguish \Fr distance at most 1 or at least 3. Therefore, Theorems~\ref{thm:bcpfrechet} and~\ref{thm:nnfrechet} even hold against multiplicative $2.999$-approximation algorithms. For approximation algorithms we refer to~\cite{BringmannM16,ChanR18}.

\item \emph{One-dimensional Curves:}
We showed hardness for curves in the plane. The same results hold for one-dimensional curves, of the form $\pi = (\pi_1,\ldots,\pi_d)$ with $\pi_i \in \mathbb{R}$, see~\cite{BringmannM16,BuchinOS19}.

\item \emph{Continuous and Weak Variants:} The same lower bounds as in Theorems~\ref{thm:bcpfrechet},~\ref{thm:nnfrechet}, and \ref{thm:frechethardness} also hold for other standard variants of the \Fr distance~\cite{Bringmann14,BuchinOS19}.

\item \emph{Realistic Input Curves:}
In order to avoid the quadratic worst-case complexity, geometers have studied several models of realistic input curves. For example, on so-called \emph{$c$-packed} curves the \Fr distance can be $(1+\eps)$-approximated in time $\tilde{O}(c n / \sqrt{\eps})$~\cite{DriemelHW12,BringmannK17}, which matches a conditional lower bound~\cite{Bringmann14}.

\item \emph{Logfactor Improvements:}
Lower bounds under OVH rule out polynomial improvements of the form $O(n^{2-\eps})$. What about logfactor improvements? An algorithm running in time $O(n^2 \log \log n / \log n)$ is known~\cite{AgarwalAKS14}. Can we improve this to time $O(n^2 / \log^{100} n)$? 
Such an improvement was shown to be unlikely, as it would imply new circuit lower bounds~\cite{AbboudB18}.
\end{itemize}

\section{More Fine-Grained Computational Geometry}

In this article we focused on nearest neighbor search and the \Fr distance. Further work on fine-grained complexity in computational geometry includes conditional lower bounds for a variant of \Fr distance between $k$ curves~\cite{buchin2016fine}, 
the dynamic time warping distance~\cite{BringmannK15,AbboudBW15},
the \Fr distance under translation~\cite{BringmannKN19} and Hausdorff distance under translation~\cite{BringmannN21},
curve simplification under \Fr distance~\cite{BringmannC19,buchin2016fine}, and
Maximum Weight Rectangle~\cite{BackursDT16}.



\bibliographystyle{splncs04}
\begin{footnotesize}

\end{footnotesize}

\end{document}